\documentclass{article}

\usepackage{amssymb}
\usepackage{amsthm}
\usepackage{graphicx}
\usepackage{subfigure}
\usepackage{wrapfig}
\usepackage{xspace}
\usepackage{times}

\usepackage[a4paper]{geometry}

\newtheorem{theorem}{Theorem}
\newtheorem{lemma}{Lemma}
\newtheorem{corollary}{Corollary}

\newtheorem{property}{Property}
\newtheorem{claim}{Claim}
\usepackage[pagewise]{lineno}

\newcommand{\coloring}{red-blue edge coloring\xspace}

\newcommand{\forestc}{forest coloring\xspace}
\newcommand{\boundedc}[1]{${#1}$-degree coloring\xspace}

\newcommand{\Forestcs}{Forest Colorings\xspace}
\newcommand{\Boundedcs}[1]{${#1}$-Degree Colorings\xspace}

\begin{document}

\title{On Partitioning the Edges of 1-Planar Graphs
\thanks{Research supported in part by the MIUR
    project AMANDA ``Algorithmics for MAssive and Networked DAta'',
    prot. 2012C4E3KT\_001.}}

\author{William J. Lenhart$^1$,
Giuseppe Liotta$^2$,
Fabrizio Montecchiani$^2$
\\[0.1in]
$^1$Department of Computer Science, Williams College, USA\\
\texttt{\small wlenhart@williams.edu}
\\
$^2$Dipartimento di Ingegneria, Universit{\`a} degli Studi di Perugia, Italy\\
\texttt{\small \{giuseppe.liotta,fabrizio.montecchiani\}@unipg.it}
}

\date{}

\maketitle

\begin{abstract}
A 1-plane graph is a graph embedded in the plane such that each edge is crossed at most once. A 1-plane graph is optimal if it has maximum edge density.  A \coloring of an optimal 1-plane graph $G$ partitions the edge set of $G$ into blue edges and red edges such that no two blue edges cross each other and no two red edges cross each other.
We prove the following:
$(i)$ Every optimal 1-plane graph has a \coloring such that the blue subgraph is maximal planar while the red subgraph 
has vertex degree at most four; this bound on the vertex degree is worst-case optimal.
$(ii)$ A \coloring may not always induce a red forest of bounded vertex degree. 
Applications of these results to graph augmentation and graph drawing are also discussed.
\end{abstract}


\section{Introduction}
A well-studied subject in graph algorithms and graph theory is how to color the edges of a graph such that each partition set induces a subgraph with special properties. For example, Colbourn and Elmallah~\cite{ec-pepg+-88} show how to compute a \coloring of a planar graph such that both the red edges and the blue edges form a partial 3-tree. A \coloring of a planar graph that partitions the edge set into partial 2-trees is proved by Kedlaya~\cite{Kedlaya1996238} and by Ding {\em et al.}~\cite{Ding2000221}. Gon{\c{c}}alves~\cite{DBLP:conf/stoc/Goncalves05} shows that every planar graph has a \coloring where each color class is an outerplanar graph and that this coloring can be computed in linear time, thus solving a famous conjecture by Chartrand, Geller, and Hedetniemi~\cite{Chartrand197112}. As for more than two colors, a classic result by Schnyder~\cite{DBLP:conf/soda/Schnyder90} is that the edges of a maximal planar graph can be colored in linear time with three colors so that each partition set forms a spanning tree of the graph, which can be used to obtain a quadratic-area planar straight-line drawing.

In this paper we study {\coloring}s of graphs that are ``almost'' planar, that is graphs for which a drawing where only some types of edge crossings are allowed. Namely, we focus on graphs that admit a drawing such that each edge is crossed by at most another edge; these graphs are called \emph{1-planar} and have been the subject of a rich body of literature partly for their applications in the area of graph drawing (see e.g.~\cite{DBLP:conf/gd/AlamBK13,DBLP:conf/wads/AlamEKPTU15,JGAA-347,DBLP:conf/compgeom/BiedlLM16,DBLP:journals/dam/BorodinKRS01,bdek+-rdicg-15,DBLP:conf/gd/BrandenburgEGGHR12,DBLP:journals/combinatorics/CzapH13,DBLP:journals/ipl/Didimo13,DBLP:journals/tcs/EadesHKLSS13,el-rac1p-DAM13,DBLP:journals/dm/FabriciM07,DBLP:conf/cocoon/HongELP12,pt-gdfce-C97,Suzuki20106,Suzuki2010,t-rdg-JGT88}).
A 1-planar graph having maximum edge density is called \emph{optimal 1-planar}. While 1-planar graphs are NP-hard to recognize~\cite{DBLP:journals/algorithmica/GrigorievB07,DBLP:journals/jgt/KorzhikM13}, optimal 1-planar graphs are recognizable in polynomial time~\cite{DBLP:journals/corr/Brandenburg15,DBLP:journals/jacm/ChenGP02}.

 Our research starts with the following observation: It is known that every optimal 1-planar graph with $n$ vertices has exactly $4n-8$ edges~\cite{bsw,bsw2,pt-gdfce-C97} and it admits a drawing with exactly $n-2$ crossings~\cite{DBLP:journals/combinatorics/CzapH13}. Therefore one can assign the red color to one edge per crossing, color all other edges as blue, and obtain two planar subgraphs. The $3n-6$ blue edges form a maximal planar graph, while the red subgraph consists of $n-2$ edges. Given the sparsity of the red subgraph, we wonder whether it can have bounded vertex degree, i.e., whether the degree of each vertex in the red subgraph can be bounded by a constant that does not depend on $n$. For example, Figure~\ref{fi:example} shows an example where the red subgraph is a forest of two paths:  one is $v_{11},v_2,v_4,v_6,v_9,v_7,$; the other path is $v_8,v_5,v_3,v_1,v_{10},v_{12}$. We recall that Czap and Hud\'ak~\cite{DBLP:journals/combinatorics/CzapH13} proved that every optimal 1-planar graph has a \coloring   such that the blue subgraph is maximal planar and the red subgraph is a forest. Ackerman~\cite{DBLP:journals/dam/Ackerman14} generalized this result to any (non-optimal) 1-planar graph. However, these results do not give bounds on the vertex degree of the red subgraph.

We prove the following theorem, where by plane (1-plane) graph we mean a planar (1-planar) embedding of the graph in the plane.




\begin{theorem}\label{th:main}
An optimal 1-plane graph with $n$ vertices admits a \coloring such that the blue edges induce a maximal plane graph and the red edges induce a 
plane graph having vertex degree at most four. This coloring can be computed in $O(n)$ time and the bound on the vertex degree of the red subgraph is worst-case optimal. Also, for any constant upper bound on the vertex degree, the red subgraph may not be a forest.
\end{theorem}

As a byproduct of the techniques in the proof of Theorem~\ref{th:main}, we prove the following result about planar augmentations of quadrangulations. It may be worth recalling that Corollary~\ref{co:main} is stated as a theorem in~\cite{{DBLP:conf/gd/FraysseixM94}} where the proof is however omitted, and that the problem of augmenting a planar graph to a triangulated planar graph by minimizing the maximum vertex degree increase is NP-complete~\cite{DBLP:journals/iandc/KantB97}.

\begin{corollary}\label{co:main}
Every 3-connected plane quadrangulation with $n$ vertices and maximum vertex degree $\Delta$ can be triangulated in $O(n)$ time such that the resulting maximal plane graph has vertex degree at most $\Delta+4$, and this bound is worst-case optimal.
\end{corollary}

A second implication of Theorem~\ref{th:main} concerns visibility representations of graphs. Visibility representations are a well-known convention to represent graphs and they have been widely studied in the literature, see e.g.~\cite{DBLP:conf/compgeom/BiedlLM16,DBLP:conf/gd/BoseDHS96,DBLP:conf/cccg/Shermer96,DBLP:conf/stacs/StreinuW03}. Each vertex is represented as an axis-aligned rectangle and each edge as an unobstructed horizontal or vertical line of sight between the corresponding pair of rectangles. Coloring blue and red the vertical and horizontal edges, respectively, we obtain a \coloring of the graph. A visibility representation is \emph{flat} if every rectangle intersects exactly one row~\cite{DBLP:conf/gd/Biedl14}, which implies that every vertex is incident to at most two red edges. The following corollary follows from Theorem~\ref{th:main}.

\begin{corollary}\label{co:visrep}
There exist optimal 1-plane graphs not admitting a flat visibility representation.
\end{corollary}


The remainder of this paper is organized as follows. Section~\ref{se:preliminaries} contains preliminaries. The proof of Theorem~\ref{th:main} is subdivided into two parts. In Section~\ref{se:decomp-forest} we show that there exist optimal 1-plane graphs not admitting any  \coloring whose red subgraph is a forest of bounded vertex degree. In Section~\ref{se:decomp-bounded} we prove that every optimal 1-plane graph has a \coloring whose red subgraph has vertex degree at most four and this bound is worst-case optimal. In Section~\ref{se:openproblems} we give our final remarks.

\section{Preliminaries}\label{se:preliminaries}

\begin{figure}[t]
\centering
\subfigure[$G$]{\includegraphics[scale=0.5,page=1]{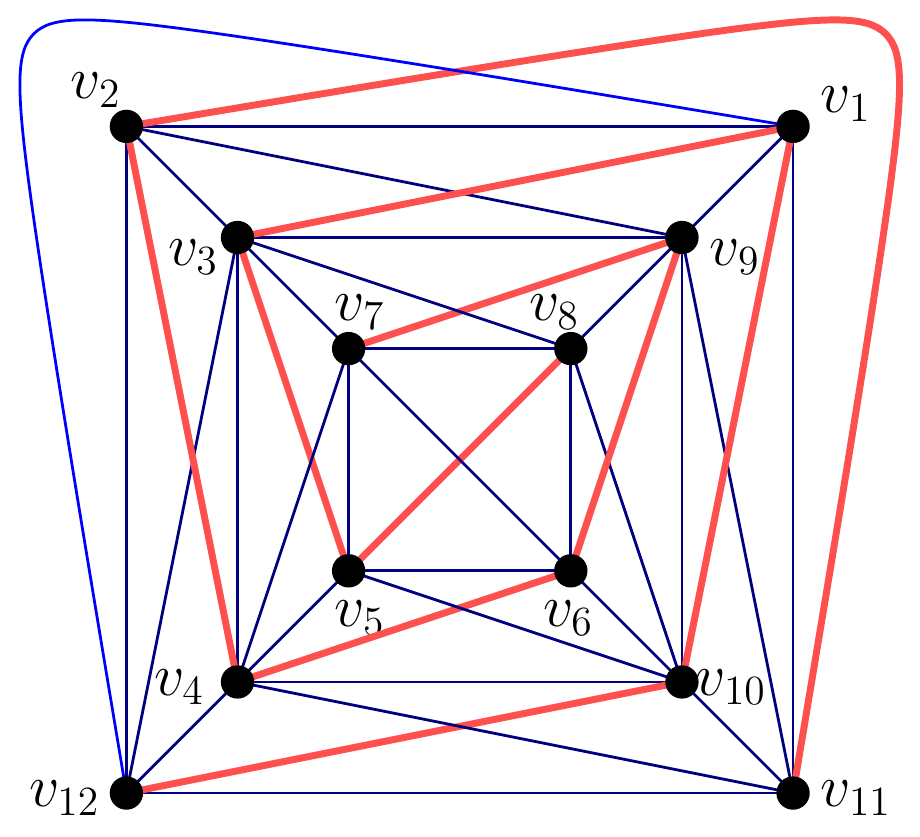}\label{fi:example}}\hfil
\subfigure[$Q(G)$]{\includegraphics[scale=0.5,page=2]{figures/optimal1planar}\label{fi:optimal1planar}}
\caption{(a) An optimal 1-plane graph $G$; the (bold) red edges, one for each crossing, induce two plane paths. (b) The subgraph $Q(G)$ of $G$, which is a 3-connected quadrangulation.}
\end{figure}

{\noindent\bf Planarity and 1-planarity.} A graph $G$ is \emph{simple}, if it contains neither loops nor multiple edges. We consider simple graphs, if not otherwise specified. A \emph{topological graph} is an embedding of a graph in the plane, where the vertices of the graph are distinct points and the edges are Jordan arcs joining the corresponding pairs of points. Also, every pair of arcs share at most one point, which is either a common end-point or an interior point where the two arcs properly cross. A topological graph in which no two edges cross is called a \emph{plane graph}. A graph that admits such a representation is a \emph{planar graph}. Furthermore, a topological graph such that every edge is crossed at most once is called a \emph{1-plane graph}, and a graph that admits such a representation is a \emph{1-planar graph}.

A plane graph divides the plane into topologically connected regions, called \emph{faces}. The unbounded region is called the \emph{outer face}. A 1-plane graph still divides the plane into faces, whose boundary may consist of edge segments between vertices and/or crossing points of edges. The unbounded region is still called the outer face.

Every optimal 1-plane graph $G$ can be obtained from a 3-connected plane quadrangulation $Q(G)$ (i.e., a 3-connected plane graph whose faces are all cycles of length four) by adding two (crossed) edges to each face of $Q(G)$~\cite{Suzuki2010}. Since $Q(G)$ is bipartite, this corresponds to coloring the vertices of $G$ with two colors, \emph{black} and \emph{white}, such that the uncrossed edges of $G$ are between vertices with different colors, while all pairs of crossing edges are between vertices of the same color. Based on this fact, we call a crossing edge having black (white) end-vertices a \emph{black (white) diagonal} of $G$. Figure~\ref{fi:optimal1planar} illustrates the subgraph $Q(G)$ of the optimal 1-plane graph $G$ in Figure~\ref{fi:example}.

\medskip

{\noindent\bf Book embeddings.} A \emph{$k$-page book embedding}, for some integer $k \geq 0$, is a particular representation of a graph $G$, such that: $(i)$ The vertices are restricted to a line, called the \emph{spine}; $(ii)$ The edges are partitioned into $k$ sets, called \emph{pages}, such that edges in a same page are drawn on a half-plane delimited by the spine and do not cross each other.

Let $G$ be an optimal 1-plane graph, and let $Q(G)$ be the plane quadrangulation of $G$. de Fraysseix {\em et al.}~\cite{DBLP:journals/dcg/FraysseixMP95} prove that $Q(G)$ has a $2$-page book embedding $D$  such that the following properties hold. We assume the vertices along the spine of $D$ to be ordered from left to right, and we call the two pages of $D$ the \emph{upper} and the \emph{lower} page.

\begin{itemize}

\item[{\bf p1.}] The leftmost and the rightmost vertices of $D$, denoted by $s_b$ and $t_b$, are both black.

\item[{\bf p2.}] The edges of the upper (lower) page induce a spanning tree of $Q(G) \setminus t_b$ ( $Q(G) \setminus s_b$).

\item[{\bf p3.}] The upper (lower) page contains edges whose leftmost end-vertex is black (white) and whose rightmost end-vertex is white (black).

\end{itemize}
Figure~\ref{fi:2pagebook} shows a $2$-page book embedding which satisfies {\bf p1--p3}.

\section{\Forestcs of Optimal 1-Plane Graphs}\label{se:decomp-forest}

A \coloring such that the blue subgraph is a maximal plane graph and the red subgraph is a forest is called a \emph{\forestc}.
In this section we show that, for any constant $c$, there exist optimal 1-plane graphs for which any \forestc is such that the red subgraph $G_R$ has maximum vertex degree at least $c$, spans all the vertices of the graph, and is composed of exactly two trees. Hence, $G_R$ has neither bounded vertex degree nor bounded size. We first show that $G_R$ spans all the vertices and is composed of two trees.

\begin{lemma}\label{le:2trees}
Let $G$ be an $n$-vertex optimal 1-plane graph with a \forestc. Then the red subgraph has $n$ vertices and it is composed of two trees.
\end{lemma}
\begin{proof}
Denote by $n$ and $m$ the number of vertices and edges of $G$, respectively. Similarly, denote by $n_R$ and $m_R = n-2$ the number of vertices and edges of the red subgraph $G_R$, respectively. Consider the plane quadrangulation $Q(G)$ of $G$ (see Section~\ref{se:preliminaries}). Recall that the edge set of $G_R$ is composed of one edge for each pair of crossing edges of $G$, i.e., it contains only black and white diagonals of $G$. Since $G$ contains at least two white and at least two black vertices, it follows that $G_R$ contains at least two trees, one induced by the white diagonals and one induced by the black diagonals. Let $T_1,\dots,T_k$ be the $k \geq 2$ trees of $G_R$, each having $n_i$ vertices. Then $m_R=\sum_{i=1}^{k}{(n_i-1)}=\sum_{i=1}^{k}{n_i}-k=n-2$. Since $\sum_{i=1}^{k}{n_i} \leq n$ and $k \geq 2$, we have that $\sum_{i=1}^{k}{n_i}-k=n-2$ holds only for $\sum_{i=1}^{k}{n_i}=n$ and $k=2$.
\end{proof}

\begin{figure}[t]
\centering
\subfigure[$G_b$]{\includegraphics[scale=0.5,page=1]{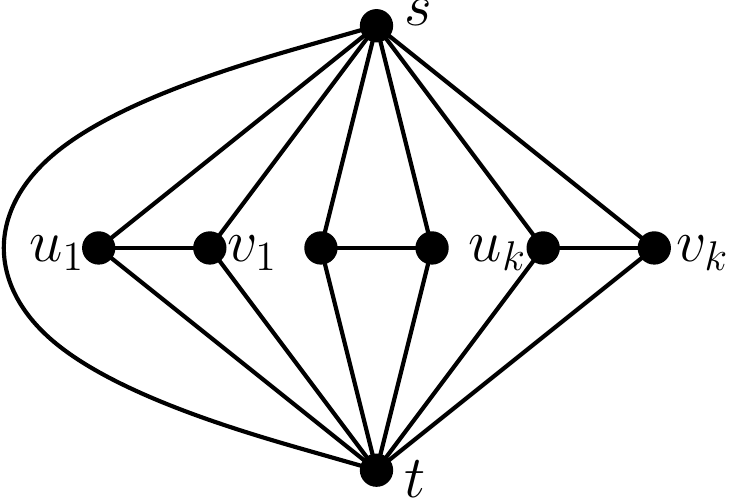}\label{fi:counterexample-1}}\hfil
\subfigure[$G^*_b$]{\includegraphics[scale=0.5,page=2]{figures/counterexample}\label{fi:counterexample-2}}\hfil
\subfigure[$G$]{\includegraphics[scale=0.5,page=3]{figures/counterexample}
\label{fi:counterexample-3}}
\caption{Illustration for the proof of Lemma~\ref{le:negative}.}
\end{figure}

To prove the negative result of Theorem~\ref{th:main}, we first construct a plane graph $G_b$ which does not admit a spanning tree of bounded degree. We then construct an optimal 1-plane graph $G$ such that $G_b$ is the subgraph of $G$ induced by its black (or, equivalently, white) diagonals. By Lemma~\ref{le:2trees}, we know that the red subgraph of any \forestc of $G$ spans all the vertices of $G$ and it is composed of exactly two trees. Thus, one of the two trees spans all the black vertices. Moreover, such a tree is a spanning tree of $G_b$, and thus cannot have bounded degree.

\medskip

Graph $G_b$ is constructed as follows. Start by adding $k$ vertex-disjoint edges $(u_1,v_1)$, $\dots$, $(u_k,v_k)$. Connect all the end-vertices of these edges to two additional vertices $s$ and $t$, and add the edge $(s,t)$; see also Figure~\ref{fi:counterexample-1}. The resulting graph $G_b$ has $n_b = 2k+2$ vertices, $m_b=5k+1$ edges, and, by Euler's formula, $f_b=3k+1$ faces. Any spanning tree $T_b$ of $G_b$ must contain either an edge connected to $s$ or an edge connected to $t$ for each of the edges $(u_1,v_1)$, $\dots$, $(u_k,v_k)$. Thus, either $s$ or $t$ must have degree at least $k/2$ in $T_b$.

Compute now the so-called \emph{angle graph} $G^*_b$ of $G_b$ by adding a (white) vertex $v_f$ inside each face $f$ of $G_b$ and connect $v_f$ to all vertices on the boundary of $f$. Graph $G^*_b$ is a triangulated plane graph by construction with $n^*_b=n_b+f_b=5k+3$ vertices and $m^*_b=3n^*_b-6=15k+3$ edges. Observe that each face of $G^*_b$ contains two black vertices and one white vertex on its boundary; see also Figure~\ref{fi:counterexample-2}. Finally, for every pair of white vertices, $v_f$ and $v_f'$, belonging to two faces that share an edge $e$ of $G_b$, add the edge $(v_f,v_f')$ and make it cross with $e$. The resulting graph $G$ is the desired optimal 1-plane graph; see also Figure~\ref{fi:counterexample-3}. Namely, $G$ is 1-plane by construction, it has $n=n^*_b=5k+3$ vertices, and $m=m^*_b+m_b=15k+3+5k+1=20k+4=4n-8$ edges (i.e., it is optimal), and its black diagonals are all and only the edges of $G_b$. We can summarize this discussion as follows.

\begin{lemma}\label{le:negative}
For any constant $c$, there exists an optimal 1-plane graph such that in any \forestc the maximum
vertex degree of the red subgraph is at least $c$.
\end{lemma}

\section{\Boundedcs{k} of Optimal 1-Plane Graphs}\label{se:decomp-bounded}

\newcommand{\algo}{\texttt{DiagPicker}\xspace}

A \coloring such that the blue subgraph is a maximal plane graph and the red subgraph has vertex degree at most $k$ (for some constant $k$) is called a \emph{\boundedc{k}}. In this Section we conclude the proof of Theorem~\ref{th:main} by showing that, if we drop the acyclicity requirement for the red subgraph, then every optimal 1-plane graph $G$ admits a \boundedc{4}. Also, we exhibit a graph not admitting a \boundedc{k} for $k<4$.

\medskip

\begin{figure}[t]
\centering
\subfigure[]{\includegraphics[scale=0.45,page=3]{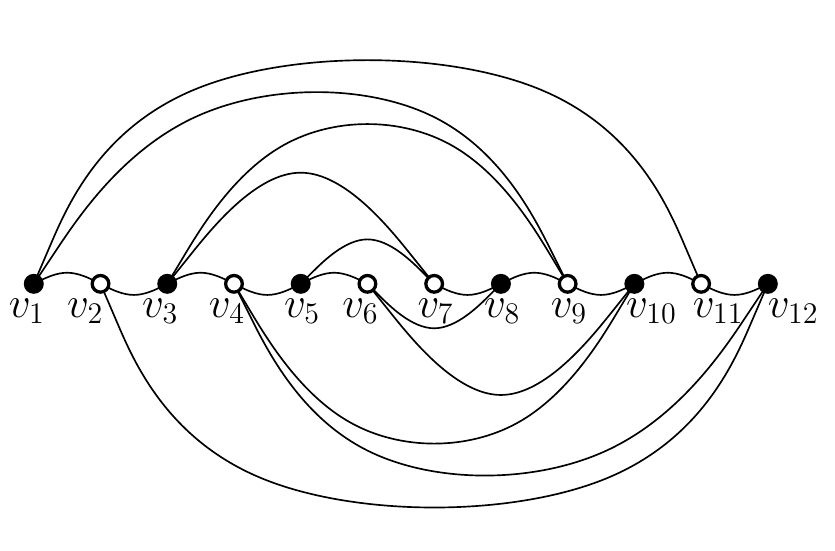}\label{fi:2pagebook}}\hfil
\subfigure[]{\includegraphics[scale=0.45,page=4]{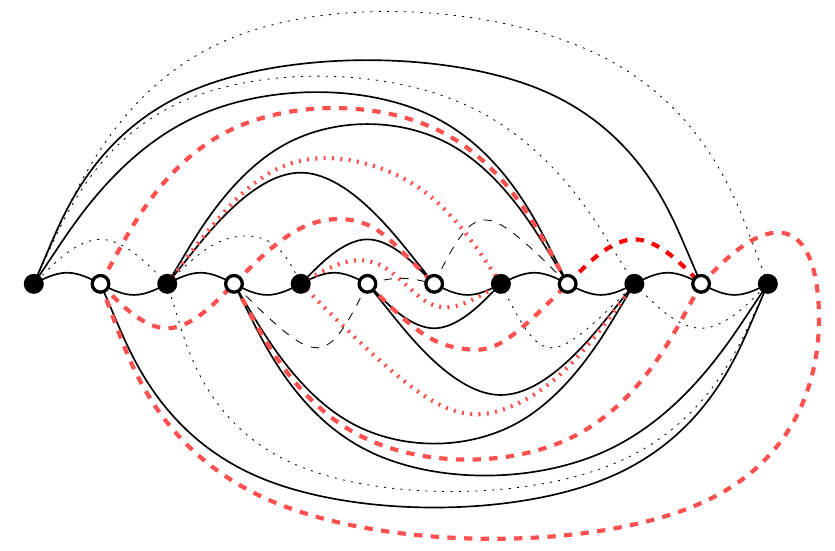}\label{fi:decomposition}}
\caption{(a) A $2$-page book embedding of $Q(G)$ (Figure~\ref{fi:optimal1planar}) computed by the algorithm in~\cite{DBLP:journals/dcg/FraysseixMP95}. (b) A partition of the edges of the optimal 1-plane graph $G$  (Figure~\ref{fi:example}), computed by algorithm \algo. The selected edges are red. }
\end{figure}

We begin by introducing some useful notation. Let $G$ be an optimal 1-plane graph, and let $Q(G)$ be the plane quadrangulation of $G$. Let $D$ be a  $2$-page book embedding of $Q(G)$ computed as in~\cite{DBLP:journals/dcg/FraysseixMP95}. Let $f$ be a face of $Q(G)$, \emph{different} from the outer face. Due to {\bf p3} of Section~\ref{se:preliminaries}, the vertices of $f$ along the spine of $D$ can be neither two whites followed by two blacks nor two blacks followed by two whites, as otherwise there would be a crossing between two edges that are on the same page. Hence, only the following four types of faces are possible in $D$. Denote by \emph{upper parachute} (resp., \emph{lower parachute}) a face  such that its four vertices are an alternated sequence of blacks (resp., whites) and whites (resp., blacks); and denote by \emph{upper dolphin} (resp., \emph{lower dolphin}) a face  such that its four vertices are a black (resp., a white) followed by two whites (resp., two blacks) and a black (resp., a white). For example, in Figure~\ref{fi:2pagebook}, all inner faces are upper and lower parachutes except for the upper dolphin $\{v_5,v_6,v_7,v_8\}$.

\medskip

We now describe an algorithm, \algo, which exploits the properties of $D$ to select one diagonal for each face of $Q(G)$, so that the graph induced by the selected diagonals has maximum vertex degree four. A \coloring such that all selected diagonals are red and all remaining edges are blue yields the desired result.

\medskip

{\bf Step 1.} The first step of the algorithm transforms $D$ as follows. Let $f$ be an upper dolphin, and let $u$ and $v$ be the two consecutive white vertices of $f$ (notice that, there is no other vertex between them). Add a (black) vertex $w$ between $u$ and $v$, and connect $w$ to $u$ and $v$. This operation splits $f$ into two faces, $f_1$ and $f_2$, such that $f_1$ is an upper parachute and $f_2$ is a lower parachute. For lower dolphins we can define an analogous operation. By splitting all upper and lower dolphins in $D$ we obtain a new $2$-page book embedding $D^*$ containing only upper and lower parachutes.


\medskip

\begin{figure}[t]
\centering
\subfigure[]{\includegraphics[scale=0.5,page=1]{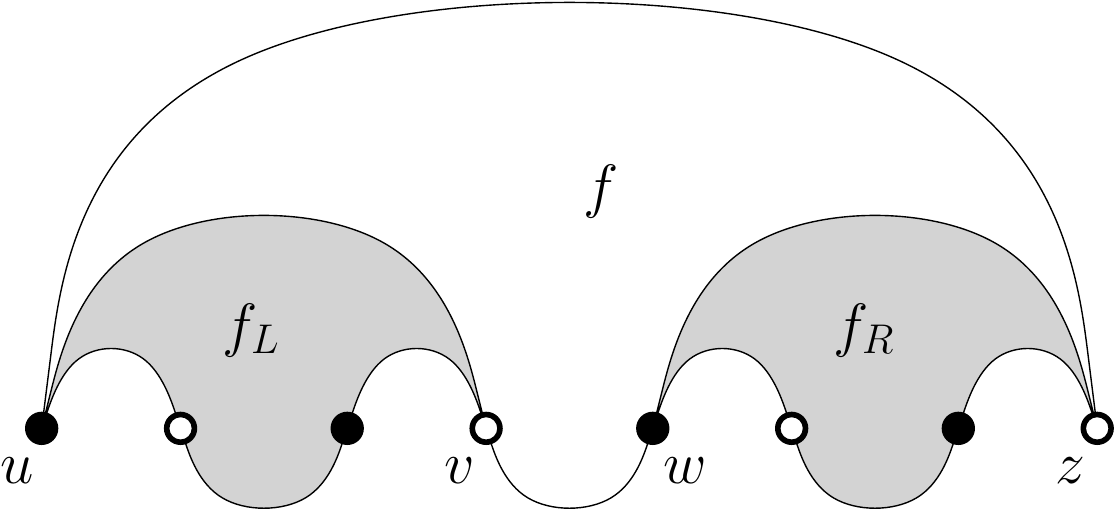}\label{fi:leftrightchild}}\hfil
\subfigure[]{\includegraphics[scale=0.5,page=2]{figures/examples}\label{fi:step3-1}}\hfil
\subfigure[]{\includegraphics[scale=0.5,page=3]{figures/examples}\label{fi:step3-2}}
\caption{(a) Illustration for the definition of left child ($f_L$) and right child ($f_R$) of an upper parachute $f$. (b)-(c) Illustration for step 3: (b) two (parallel) white diagonals have been chosen; (b) two black diagonals have been chosen. }
\end{figure}

{\bf Step 2.} In the second step, \algo selects one diagonal per parachute according to the following rule. We first focus on the upper parachutes, and define some further notation.  Let $f$ be an upper parachute of $D^*$ and let $u,v,w,z$ be its four vertices in the left to right order they appear along the spine. Let $f_L$ be the upper parachute (if any) distinct from $f$ and sharing the edge $(u,v)$ with $f$. We call $f_L$ the \emph{left child} of $f$. Similarly, we call the \emph{right child} of $f$, the upper parachute (if any) $f_R$ sharing the edge $(w,z)$ with $f$. Furthermore, $f$ is called the \emph{parent} of $f_L$ and $f_R$. See also Fig.~\ref{fi:leftrightchild} for an illustration.

Let $f$ be an upper parachute. If $f$ has no parent or $f$ is a left child, then the white diagonal of $f$ is selected. Otherwise, if $f$ is a right child, then its black diagonal is selected; see also Figure~\ref{fi:decomposition} for an example.  Let $S$ be the set of selected edges, one for each upper parachute. We prove the following lemma.

\begin{lemma}\label{le:upperpage}
Let $G_S$ be the graph induced by the edges in $S$. Then $G_S$ is a forest of paths.
\end{lemma}
\begin{proof}
To prove the claim, observe first that the parent-child relationship between upper parachutes yields a forest of binary rooted trees, denoted by $\mathcal F$. Namely, if an upper parachute $f$ has $f_L$ and $f_R$ as left and right children, then the corresponding node $n_f$ has $n_{f_L}$ and $n_{f_R}$ as left and right children in a tree of $\mathcal F$. Let $T$ be a tree of $\mathcal F$. Let $n_f$ be a node of $T$ corresponding to an upper parachute $f$ of $D^*$. Let $S_f$ be the subset of $S$ containing only the selected diagonals for the parachutes in the subtree $T_f$ of $T$ rooted at $n_f$. We visit $T$ bottom up, and, for each node $n_f$ of $T$, we show that the following invariant holds: {\bf (I).} The edges in $S_f$ induce a forest of  paths. Let $n_f$ be a leaf of $T$, then Invariant {\bf (I)} clearly holds.

\begin{claim}
Let $n_f$ be a node of $T$. Let $n_{f_L}$ and $n_{f_R}$ be its two children (one of them may not exist), for which Invariant {\bf (I)} holds. Then Invariant {\bf (I)} holds for $n_f$.
\end{claim}
\begin{proof}
Suppose the white diagonal $e=(v,z)$ of $f$ is selected by \algo. Let $v$ be the leftmost end-vertex of $e$. Let the right boundary of $T_{f_L}$ (resp., $T_{f_R}$) be the path obtained starting from $n_{f_L}$ (resp., $n_{f_R}$) and by iteratively adding to the path the right child of the last added node, until the last added node has no right child. Vertex $v$ belongs to $f$ and to the set of upper parachutes that correspond to the nodes on the right boundary of $T_{f_L}$. It follows that $v$ is adjacent to exactly one edge in $S_{f_L}$ (the white diagonal of $f_L$). Similarly, $z$ belongs to $f$ and to the set of upper parachutes that correspond to the nodes on the right boundary of $T_{f_R}$. Hence, $z$ is adjacent to no edge in $S_{f_R}$. Thus, the set $S = S_{f_L} \cup S_{f_R} \cup \{e\}$ satisfies the claim. The argument is symmetric in case \algo selected the black diagonal of $f$.
\end{proof}

To conclude the proof of the lemma, we observe that each tree of $\mathcal F$ is such that the upper parachute corresponding to its root node does not share any vertex with the upper parachutes of any other tree, hence,  $G_S$ is a forest of paths.
\end{proof}

By rotating $D^*$ of an angle $-\pi$, the lower parachutes become upper parachutes, and algorithm \algo can be applied so to select one diagonal per upper parachute in the rotated drawing. It follows that, by Lemma~\ref{le:upperpage}, the final set $S$ of selected edges (one for each upper and lower parachute), induces a graph $G_S$ of maximum vertex degree four.

\medskip

{\bf Step 3.} In the third step, algorithm \algo removes the dummy vertices inserted during Step 1, without increasing the maximum vertex degree of the graph $G_S$. Let $f$ be an upper (resp., lower) dolphin in $D$, and let $f_1$ and $f_2$ be the two parachutes generated by adding a dummy vertex as explained in Step 1. If, for at least one of $f_1$ and $f_2$, a white (resp., black) diagonal $e$ is in $G_S$, then we can keep $e$ in $G_S$ and remove the other diagonal, which can be a parallel white (resp., black) diagonal between the same pair of vertices or a black (resp., white) diagonal. See also Fig.~\ref{fi:step3-1} for an illustration. If for both $f_1$ and $f_2$ the black (resp., white) diagonals are in $G_S$, then we can remove both of them and add the black (resp., white) diagonal of $f$ instead. See also Fig.~\ref{fi:step3-2} for an illustration.

We conclude the description of this step by showing a property of $G_S$ that will be used in the last step. Observe that, so far, the algorithm never selected two diagonals which share an end-vertex $u$ and such that the opposite end-vertices are both to the right or both to the left of $u$. Hence, the following holds. 

\begin{property}\label{pr:extremeverts}
Let $G'$ be the subgraph of $G_S$ induced by a subset of vertices. Then the leftmost and the rightmost vertices of $G'$ have degree at most two in $G'$.
\end{property}

\medskip

{\bf Step 4.} We conclude by showing how \algo chooses one diagonal for the outer face of $Q(G)$ to put in $G_S$ (after this addition Property~\ref{pr:extremeverts} may not hold anymore). 
%
By {\bf p1} of Section~\ref{se:preliminaries}, the leftmost and the rightmost vertices of $D$, $s_b$ and $t_b$, are both black, and they clearly belong to the outer face of $Q(G)$. By Property~\ref{pr:extremeverts}, they both have degree at most two. Thus, we might select the black diagonal that connects $s_b$ and $t_b$, and the maximum vertex degree of $G_S$ would not exceed four. However, the white diagonal of the outer face of $Q(G)$ can be selected as well. Namely, by {\bf p2} of Section~\ref{se:preliminaries}, the first vertex after $s_b$, denoted by $v_2$, and the first vertex before $t_b$, denoted by $v_{n-1}$, are also on the outer face of $Q(G)$, and thus they are both white. Moreover, $v_2$ and $v_{n-1}$ have degree at most three, since $v_2$ (resp., $v_{n-1}$) has at most two incident edges whose opposite end-vertex is to the left (resp., right) by Property~\ref{pr:extremeverts}, and at most one edge whose opposite end-vertex is to the right (resp., left). Hence, we can select the white diagonal connecting $v_2$ and $v_{n-1}$, and the maximum vertex degree of $G_S$ does not exceed four.  Figure~\ref{fi:decomposition} shows the output of \algo applied to the graph in Figure~\ref{fi:example}. 

We can summarize this discussion as follows.

\begin{lemma}\label{le:positive}
Every optimal 1-plane graph $G$ with $n$ vertices admits a \boundedc{4} that can be computed in $O(n)$ time. 
\end{lemma}
\begin{proof}
We already shown that algorithm \algo correctly computes a \boundedc{4} of $G$. In terms of time complexity, the subgraph $Q(G)$ can be extracted from $G$ in $O(n)$ time by only picking the uncrossed edges. A $2$-page book embedding of $Q(G)$ which satisfies {\bf p1--p3} can also be computed in $O(n)$ time~\cite{DBLP:journals/dcg/FraysseixMP95}. Since $Q(G)$ has $O(n)$ faces, algorithm \algo runs in $O(n)$ time. 
\end{proof}

\begin{figure}[t]
\centering
\subfigure[$H$]{\includegraphics[scale=0.5,page=1]{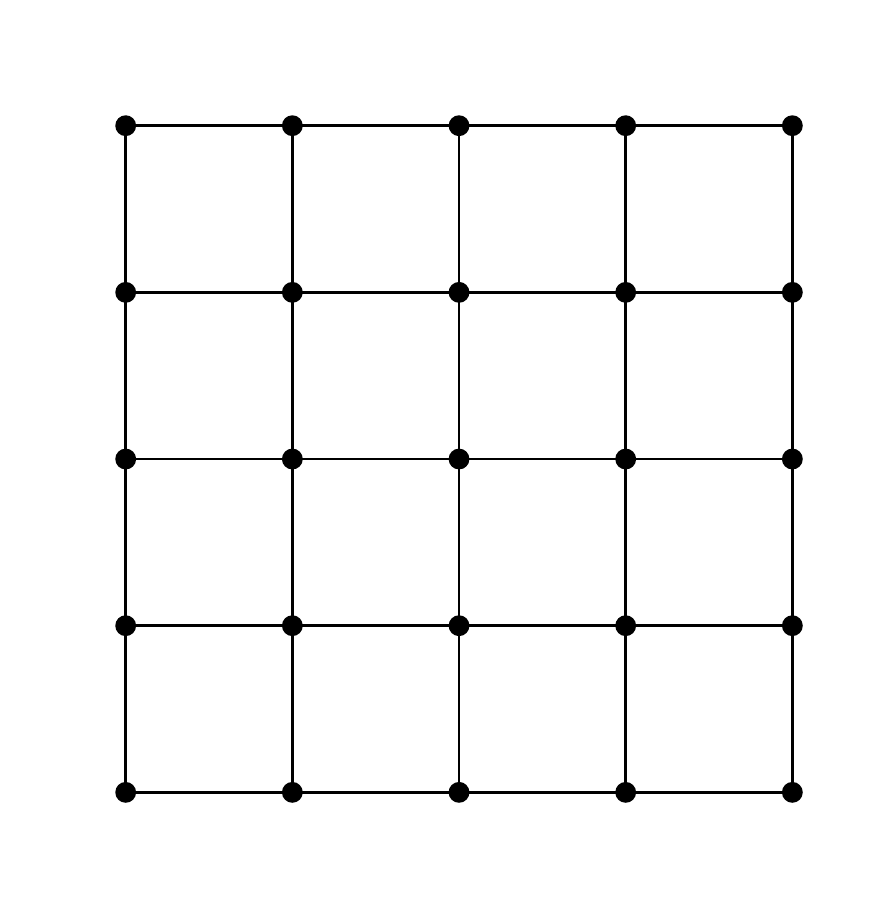}\label{fi:constr3-1}}
\subfigure[$H^*$]{\includegraphics[scale=0.5,page=2]{figures/worstcase}\label{fi:constr3-2}}
\subfigure[$Q$]{\includegraphics[scale=0.5,page=3]{figures/worstcase}\label{fi:constr3-3}}
\caption{Illustration for the proof of Lemma~\ref{le:worstcase}.}
\end{figure}

\medskip


To conclude the proof of Theorem~\ref{th:main}, we shall exhibit an optimal 1-plane graph which does not admit a \boundedc{k} for $k<4$. 
Consider a plane graph $H$ consisting of a $h \times h$ grid, see e.g. Fig.~\ref{fi:constr3-1} for $h=5$. We call \emph{black} the vertices of $H$. Add a 4-cycle of \emph{gray} vertices inside each face of $H$, except for the outer face, and connect each gray vertex to a black vertex without introducing crossings, as shown in Fig.~\ref{fi:constr3-2}. We call $H^*$ the resulting plane graph. Note that $H^*$ is 3-connected and all faces have degree four, except for the outer face. Add edges to the outer face of $H^*$, such that the resulting graph $Q$ is a 3-connected quadrangulation. Since the outer face of $H^*$ has length $4h-4$, this can be done, for example, by iteratively connecting sequences of vertices of length four along the outer face, until such a face consists of four vertices. See for example Fig.~\ref{fi:constr3-3}. Consider now the optimal 1-plane graph $G$ obtained by adding a pair of crossing edges inside every face of $Q$. Graph $Q$ has $n=h^2+4(h-1)^2=5h^2-8h+4$ vertices and, by Euler's formula, $f=n-2=5h^2-8h+2$ faces. Consider any \boundedc{k} of $G$. The red subgraph $G_R$ contains one diagonal for each face of $Q$, i.e., it has $m_R = 5h^2-8h+2$ edges. By the handshaking lemma $\sum_{v \in G_R} deg(v) = 2m_R=10h^2-16h+4$. Since each 4-cycle of gray vertices is inside a 4-cycle of black vertices, one can see that the sum of the degrees of each gray 4-cycle in the red subgraph is exactly $6$. Also, since the degree of the $h^2$ black vertices in the red subgraph is at most $k$ by assumption, and since we have $(h-1)^2$ gray 4-cycles, we have that $\sum_{v \in G_R} deg(v) = 10h^2-16h+4 \leq 6(h-1)^2 + kh^2$, which implies that $k \geq 4-\frac{4}{h}-\frac{2}{h^2} > 3$, for $h \geq 5$, as desired.

\begin{lemma}\label{le:worstcase}
For every $h\geq 5$, there exists an optimal 1-plane graph $G_h$ with $5h^2-8h+4$ vertices, such that $G_h$ does not admit any \boundedc{k} with $k < 4$.
\end{lemma}

Lemmas~\ref{le:negative},~\ref{le:positive} and~\ref{le:worstcase} prove Theorem~\ref{th:main}.

We recall that any 3-connected plane quadrangulation can serve as the subgraph $Q(G)$
of an optimal 1-plane graph $G$~\cite{Suzuki2010}, and that the red edges of a \boundedc{4} of $G$ can be used to triangulate $Q(G)$. We therefore have Corollary~\ref{co:main} as a byproduct of Lemmas~\ref{le:positive} and~\ref{le:worstcase}.

Finally, Lemma~\ref{le:worstcase} implies Corollary~\ref{co:visrep}.

\section{Final Remarks}\label{se:openproblems}
We proved that it is not possible to pick an edge for each pair of crossing edges of an optimal 1-plane graph $G$ such that the graph induced by the selected edges is a forest of bounded degree. Motivated by this negative result, we showed that if we drop the acyclicity requirement, then an edge decomposition into a triangulated plane graph and a 
plane graph of maximum vertex degree four can be computed in linear time, and that this bound on the vertex degree is worst-case optimal. It remains open to establish whether every 1-plane graph admits a \boundedc{k}  for some constant $k$.

\section*{Acknowledgments}
We thank Therese Biedl for useful discussions on visibility representations of non-planar graphs.

{\small \bibliography{optimal1p}}

\begin{thebibliography}{10}

\bibitem{DBLP:journals/dam/Ackerman14}
E.~Ackerman.
\newblock A note on 1-planar graphs.
\newblock {\em Discrete Appl. Math.}, 175:104--108, 2014.

\bibitem{DBLP:conf/gd/AlamBK13}
M.~J. Alam, F.~Brandenburg, and S.~G. Kobourov.
\newblock Straight-line grid drawings of 3-connected 1-planar graphs.
\newblock In {\em {GD} 2013}, volume 8242 of {\em LNCS}, pages 83--94.
  Springer, 2013.

\bibitem{DBLP:conf/wads/AlamEKPTU15}
M.~J. Alam, W.~S. Evans, S.~G. Kobourov, S.~Pupyrev, J.~Toeniskoetter, and
  T.~Ueckerdt.
\newblock Contact representations of graphs in {3D}.
\newblock In {\em {WADS} 2015}, volume 9214 of {\em LNCS}, pages 14--27.
  Springer, 2015.

\bibitem{JGAA-347}
C.~{Auer}, F.~{Brandenburg}, A.~{Gleißner}, and J.~{Reislhuber}.
\newblock 1-planarity of graphs with a rotation system.
\newblock {\em J. Graph Algorithms Appl.}, 19(1):67--86, 2015.

\bibitem{DBLP:conf/gd/Biedl14}
T.~C. Biedl.
\newblock Height-preserving transformations of planar graph drawings.
\newblock In {\em {GD} 2014}, volume 8871 of {\em LNCS}, pages 380--391.
  Springer, 2014.

\bibitem{DBLP:conf/compgeom/BiedlLM16}
T.~C. Biedl, G.~Liotta, and F.~Montecchiani.
\newblock On visibility representations of non-planar graphs.
\newblock In {\em {SoCG} 2016}, volume~51 of {\em LIPIcs}, pages 19:1--19:16.
  Schloss Dagstuhl - Leibniz-Zentrum fuer Informatik, 2016.

\bibitem{bsw}
R.~Bodendiek, H.~Schumacher, and K.~Wagner.
\newblock Bemerkungen zu einem sechsfarbenproblem von g. ringel.
\newblock {\em Abh. aus dem Math. Seminar der Univ. Hamburg}, 53:41--52, 1983.

\bibitem{bsw2}
R.~Bodendiek, H.~Schumacher, and K.~Wagner.
\newblock {\"U}ber 1-optimale graphen.
\newblock {\em Mathematische Nachrichten}, 117:323–--339, 1984.

\bibitem{DBLP:journals/dam/BorodinKRS01}
O.~V. Borodin, A.~V. Kostochka, A.~Raspaud, and E.~Sopena.
\newblock Acyclic colouring of 1-planar graphs.
\newblock {\em Discrete Appl. Math.}, 114(1-3):29--41, 2001.

\bibitem{DBLP:conf/gd/BoseDHS96}
P.~Bose, A.~M. Dean, J.~P. Hutchinson, and T.~C. Shermer.
\newblock On rectangle visibility graphs.
\newblock In {\em {GD} 1996}, volume 1190 of {\em LNCS}, pages 25--44.
  Springer, 1997.

\bibitem{DBLP:journals/corr/Brandenburg15}
F.~Brandenburg.
\newblock On 4-map graphs and 1-planar graphs and their recognition problem.
\newblock {\em CoRR}, abs/1509.03447, 2015.

\bibitem{bdek+-rdicg-15}
F.~Brandenburg, W.~Didimo, W.~S. Evans, P.~Kindermann, G.~Liotta, and
  F.~Montecchiani.
\newblock Recognizing and drawing {IC}-planar graphs.
\newblock In {\em {GD} 2015}, LNCS. Springer, 2015.
\newblock to appear.

\bibitem{DBLP:conf/gd/BrandenburgEGGHR12}
F.~Brandenburg, D.~Eppstein, A.~Glei{\ss}ner, M.~T. Goodrich, K.~Hanauer, and
  J.~Reislhuber.
\newblock On the density of maximal 1-planar graphs.
\newblock In {\em {GD} 2012}, volume 7704 of {\em LNCS}, pages 327--338.
  Springer, 2012.

\bibitem{Chartrand197112}
G.~Chartrand, D.~Geller, and S.~Hedetniemi.
\newblock Graphs with forbidden subgraphs.
\newblock {\em J. Combin. Theory Ser. B}, 10(1):12 -- 41, 1971.

\bibitem{DBLP:journals/jacm/ChenGP02}
Z.~Chen, M.~Grigni, and C.~H. Papadimitriou.
\newblock Map graphs.
\newblock {\em J. {ACM}}, 49(2):127--138, 2002.

\bibitem{DBLP:journals/combinatorics/CzapH13}
J.~Czap and D.~Hud{\'{a}}k.
\newblock On drawings and decompositions of 1-planar graphs.
\newblock {\em Electr. J. Comb.}, 20(2):P54, 2013.

\bibitem{DBLP:conf/gd/FraysseixM94}
H.~{de Fraysseix} and P.~{Ossona de Mendez}.
\newblock Regular orientations, arboricity, and augmentation.
\newblock In {\em {GD} 1994}, volume 894 of {\em LNCS}, pages 111--118.
  Springer, 1994.

\bibitem{DBLP:journals/dcg/FraysseixMP95}
H.~{de Fraysseix}, P.~{Ossona de Mendez}, and J.~Pach.
\newblock A left-first search algorithm for planar graphs.
\newblock {\em Discrete {\&} Comput. Geom.}, 13:459--468, 1995.

\bibitem{DBLP:journals/ipl/Didimo13}
W.~Didimo.
\newblock Density of straight-line 1-planar graph drawings.
\newblock {\em Inform. Process. Lett.}, 113(7):236--240, 2013.

\bibitem{Ding2000221}
G.~Ding, B.~Oporowski, D.~P. Sanders, and D.~Vertigan.
\newblock Surfaces, tree-width, clique-minors, and partitions.
\newblock {\em J. Combin. Theory Ser. B}, 79(2):221 -- 246, 2000.

\bibitem{DBLP:journals/tcs/EadesHKLSS13}
P.~Eades, S.~Hong, N.~Katoh, G.~Liotta, P.~Schweitzer, and Y.~Suzuki.
\newblock A linear time algorithm for testing maximal 1-planarity of graphs
  with a rotation system.
\newblock {\em Theor. Comput. Sci.}, 513:65--76, 2013.

\bibitem{el-rac1p-DAM13}
P.~Eades and G.~Liotta.
\newblock Right angle crossing graphs and 1-planarity.
\newblock {\em Discrete Appl. Math.}, 161(7-8):961--969, 2013.

\bibitem{ec-pepg+-88}
E.~S. Elmallah and C.~J. Colbourn.
\newblock Partitioning the edges of a planar graph into two partial k-trees.
\newblock In {\em Congressus Numerantium}, pages 69--80, 1988.

\bibitem{DBLP:journals/dm/FabriciM07}
I.~Fabrici and T.~Madaras.
\newblock The structure of 1-planar graphs.
\newblock {\em Discrete Math.}, 307(7-8):854--865, 2007.

\bibitem{DBLP:conf/stoc/Goncalves05}
D.~Gon{\c{c}}alves.
\newblock Edge partition of planar sraphs into two outerplanar graphs.
\newblock In {\em {STOC} 2005}, pages 504--512. {ACM}, 2005.

\bibitem{DBLP:journals/algorithmica/GrigorievB07}
A.~Grigoriev and H.~L. Bodlaender.
\newblock Algorithms for graphs embeddable with few crossings per edge.
\newblock {\em Algorithmica}, 49(1):1--11, 2007.

\bibitem{DBLP:conf/cocoon/HongELP12}
S.~Hong, P.~Eades, G.~Liotta, and S.~Poon.
\newblock F{\'{a}}ry's theorem for 1-planar graphs.
\newblock In {\em {COCOON} 2012}, volume 7434 of {\em LNCS}, pages 335--346.
  Springer, 2012.

\bibitem{DBLP:journals/iandc/KantB97}
G.~Kant and H.~L. Bodlaender.
\newblock Triangulating planar graphs while minimizing the maximum degree.
\newblock {\em Inf. Comput.}, 135(1):1--14, 1997.

\bibitem{Kedlaya1996238}
K.~S. Kedlaya.
\newblock Outerplanar partitions of planar graphs.
\newblock {\em J. Combin. Theory Ser. B}, 67(2):238 -- 248, 1996.

\bibitem{DBLP:journals/jgt/KorzhikM13}
V.~P. Korzhik and B.~Mohar.
\newblock Minimal obstructions for 1-immersions and hardness of 1-planarity
  testing.
\newblock {\em J. Graph Theory}, 72(1):30--71, 2013.

\bibitem{pt-gdfce-C97}
J.~Pach and G.~T\'oth.
\newblock Graphs drawn with few crossings per edge.
\newblock {\em Combinatorica}, 17(3):427--439, 1997.

\bibitem{DBLP:conf/soda/Schnyder90}
W.~Schnyder.
\newblock Embedding planar graphs on the grid.
\newblock In D.~S. Johnson, editor, {\em {SODA} 1990}, pages 138--148. {SIAM},
  1990.

\bibitem{DBLP:conf/cccg/Shermer96}
T.~C. Shermer.
\newblock On rectangle visibility graphs. {III.} external visibility and
  complexity.
\newblock In {\em {CCCG} 1996}, pages 234--239. Carleton University Press,
  1996.

\bibitem{DBLP:conf/stacs/StreinuW03}
I.~Streinu and S.~Whitesides.
\newblock Rectangle visibility graphs: Characterization, construction, and
  compaction.
\newblock In {\em {STACS} 2003}, volume 2607 of {\em LNCS}, pages 26--37.
  Springer, 2003.

\bibitem{Suzuki20106}
Y.~Suzuki.
\newblock Optimal 1-planar graphs which triangulate other surfaces.
\newblock {\em Discrete Math.}, 310(1):6 -- 11, 2010.

\bibitem{Suzuki2010}
Y.~Suzuki.
\newblock Re-embeddings of maximum 1-planar graphs.
\newblock {\em SIAM J. on Discrete Math.}, 24(4):1527--1540, 2010.

\bibitem{t-rdg-JGT88}
C.~Thomassen.
\newblock Rectilinear drawings of graphs.
\newblock {\em J. Graph Theory}, 12(3):335--341, 1988.

\end{thebibliography}
\bibliographystyle{abbrv}
\end{document}